\documentclass{conm-p-l}

\usepackage{amssymb,amsmath,euscript}

\usepackage{}

\newtheorem{theorem}{Theorem}[section]
\newtheorem{lemma}[theorem]{Lemma}
\newtheorem{proposition}[theorem]{Proposition}

\newtheorem{corollary}[theorem]{Corollary}

\theoremstyle{definition}

\theoremstyle{remark}
\newtheorem{remark}[theorem]{Remark}
\def\Fq{{\mathbb F}_q}
\def\a{{\alpha}}
\def\d{{\delta}}
\def\AA{{\mathbb A}}

\def\PP{{\mathbb P}}

\newcommand{\flm}{\EuScript{F}(\ell,m)}
\newcommand{\flmh}{\EuScript{F}(\ell,m;h)}

\newcommand{\C}{C^\mathbb{A}(\ell,m)}
\newcommand{\Ch}{C^\mathbb{A}(\ell,m;h)}
\newcommand{\Cl}{C^\mathbb{A}(\ell,m;\ell)}

\newcommand{\Ev}{\operatorname{Ev}}
\newcommand{\codim}{\operatorname{codim}}
\newcommand{\RM}{\mathrm{RM}}
\newcommand{\PRM}{\mathrm{PRM}}
\newcommand{\GL}{\mathrm{GL}}
\newcommand{\supp}{\mathrm{Supp}}

\def\A{{\mathbb A}}
\def\PP{{\mathbb P}}

\def\M{{\mathcal M}}
\def\N{{\mathcal N}}
\newcommand{\Lh}{\mathcal{L}_{h}}

\newcommand{\w}{{\mathrm{w_H}}}
\begin{document}

\title[Higher Weights of Affine Grassmann Codes and Their Duals]{Higher Weights of Affine Grassmann Codes and Their Duals\footnote{To appear in the Proceedings of AGCT-2013 (Luminy, France) in:
``Algorithmic Arithmetic, Geometry and Coding Theory'', \emph{Contemporary Mathematics} Vol. 637, Amer. Math. Soc.,  2015.}} %Providence,


%    author one information
\author{Mrinmoy Datta}
\address{Department of Mathematics,
Indian Institute of Technology Bombay,\newline \indent
Powai, Mumbai 400076, India.}
%\curraddr{}
\email{mrinmoy.dat@gmail.com}
\thanks{The first named author is partially supported by a doctoral fellowship from the National Board for Higher Mathematics, a division of the Department of Atomic Energy, Govt. of India.}

%    author two information
\author{Sudhir R. Ghorpade}
\address{Department of Mathematics, 
Indian Institute of Technology Bombay,\newline \indent
Powai, Mumbai 400076, India.}
\email{srg@math.iitb.ac.in}
\thanks{The second named author is partially supported by Indo-Russian project INT/RFBR/P-114 from the Department of Science \& Technology, Govt. of India and  IRCC Award grant 12IRAWD009 from IIT Bombay.}

%\subjclass[2010]{11T35, 11T06 20G40, 15B05}
%    The 2010 edition of the Mathematics Subject Classification is
%    now available.  If you are citing a classification from the
%    new scheme, use the following input coding instead.
%\subjclass[2010]{Primary }

\date{}

\begin{abstract}
We consider the question of determining the higher weights or the generalized Hamming weights 
of affine Grassmann codes and their duals. Several initial as well as terminal higher weights of affine Grassmann codes of an arbitrary level are determined explicitly. In the case of duals of these codes, we give a formula for many initial as well as terminal higher weights. As a special case, we obtain an alternative simpler proof of the formula of Beelen et al for the minimum distance of the dual of an affine Grasmann code. 
\end{abstract}

\maketitle


\section{Introduction}
\label{sec:in}
A $q$-ary linear code of length $n$ and dimension $k$, or in short, a $[n,k]_q$-code, is simply a $k$-dimensional subspace of the $n$-dimensional vector space $\Fq^n$ over the finite field $\Fq$ with $q$ elements. A basic example is that of a (generalized) Reed-Muller code $\RM (\nu, \delta)$ of order $\nu$ and length $n:=q^{\delta}$,
given by the image of the evaluation map 
$$
\Ev: \Fq[X_1, \dots X_{\delta}]_{\le \nu} \to \Fq^n \quad \text{ defined by } \quad 
\Ev (f) = \left( f(P_1), \dots , f(P_n)\right),
$$
where $\Fq[X_1, \dots X_{\delta}]_{\le \nu} $ denotes the space of polynomials in $\delta$ variables of (total) degree $\le \nu$ with coefficients in $\Fq$ and $P_1, \dots , P_n$ is 
an ordered listing %a listing 
of the points of the affine space $\A^{\delta}(\Fq) = \Fq^{\delta}$. A useful variant of this is the projective Reed-Muller code $\PRM (\nu, \delta)$ of order $\nu$ and length $n:=(q^{\delta+1}-1)/(q-1)$, which is obtained by evaluating homogeneous polynomials 
in $\delta+1$ variables of  degree $\nu$ with coefficients in $\Fq$ at points of the projective space $\PP^{\delta} = \PP^{\delta}(\Fq)$ or rather at suitably normalized representatives in $\Fq^{\delta+1}$ 
of an ordered listing of the points of $\PP^{\delta}$. 

{F}rom a geometric viewpoint, projective Reed-Muller codes $\PRM (\nu, \delta)$ correspond (at least when $\nu <q$) to the Veronese variety given by the image of $\PP^{\delta}$ in $\PP^{k-1}$ under the Veronese map of degree $\nu$, where $k:= {{\nu + \delta}\choose{\nu} }$. In this set-up, $\RM (\nu, \delta)$ corresponds to the image of this Veronese map when restricted to an $\A^{\delta}$ inside $\PP^{\delta}$ (for instance, the set of points $(x_0:x_1:\cdots : x_{\delta})$ of $\PP^{\delta}$ with $x_0=1$). 

%The study of 
Reed-Muller codes are classical objects and in the generalized setting above, their study goes back at least to  Kasami, Lin, and Peterson \cite{KLP} as well as Delsarte,  Goethals, and MacWilliams \cite{DGM}. One may refer to \cite[Prop. 4]{BGH2} for a summary of several of the basic properties of $\RM (\nu, \delta)$. Projective Reed-Muller codes appeared explicitly in the work of Lachaud \cite{L1, L2} and S{\o}rensen \cite{So}. Around the same time, a new class of codes called Grassmann codes were studied by Ryan \cite{R1, R2}, and later by Nogin \cite{N} and several others (see, e.g., \cite{GL,GPP,HJR,GK}).  These correspond geometrically to the Grassmann variety $G_{\ell,m}$ formed by the $\ell$-dimensional subspaces of $\Fq^m$ together with the Pl\"ucker embedding $G_{\ell,m} \hookrightarrow \PP^{k-1}$, where $k= {{m}\choose{\ell}}$. In effect,  the Grassmann code $C(\ell,m)$ is a %n $[n,k]_q$-
linear code whose generator matrix has  as its columns certain fixed 
representatives in $\Fq^k$ of the Pl\"ucker coordinates of all $\Fq$-rational points of $G_{\ell,m}$. Affine Grassmann codes were introduced in \cite{BGH1} and further studied in \cite{BGH2} and \cite{GK}. Given positive integers $\ell, \ell'$ with $\ell \le \ell'$, upon letting 
$m= \ell + \ell'$ and  $\delta = \ell \ell'$,  the affine Grassmann code $C^{\A}(l,m)$ is defined, like a Reed-Muller code, as the $q$-ary linear code 
 of length $n= q^{\delta}$ given by the image of the evaluation map 
\begin{equation}
\label{EvonFlm}
\Ev: \flm \to \Fq^n \quad \text{ defined by } \quad 
\Ev (f) = \left( f(P_1), \dots , f(P_n)\right),
\end{equation}
where $\flm$ is the space of linear polynomials in the minors of a generic $ \ell \times \ell'$ matrix $X$ and $P_1, \dots, P_n$ is an ordered listing of the $\delta$-dimensional affine space of all $ \ell \times \ell'$ matrices with entries in $\Fq$.  The relationship between affine Grassmann codes $C^{\A}(l,m)$ and Grassmann codes $C(l,m)$ is akin to that between Reed-Muller codes $\RM(\nu,\delta)$ and projective Reed-Muller codes $\PRM(\nu,\delta)$. 

The notion of higher weight, also known as generalized Hamming weight, of a linear code is a natural and useful generalization of the basic notion of minimum distance (cf. \cite{W}).  If $C$ is a $[n,k]_q$-code, then for $r=0,1, \dots ,k$, the $r^{{\rm th}}$ \emph{higher weight} of  $C$ is
defined by  
$$
d_r = d_r (C) = \min \{ \w( D) : D \mbox{ is a subspace of $C$ with } 
\dim D = r\},
$$
where $\w (D)$ denotes the support weight %(Haming) weight 
of $D$ [see Section \ref{sec2} below for a definition]. 
%, i.e., the number of $i\in\{1, \dots , n\}$ such that there is some $c=(c_1, \dots , c_n)\in D$ %with $c_i\ne 0$. 
Clearly, $d_1(C)$ is the minimum distance $d(C)$ of $C$. It is well-known and easy to see that $0=d_0<d_1<\cdots < d_k$ and moreover $d_k=n$ provided $C$ is nondegenerate. It is, in general, an interesting and difficult question to determine the weight hierarchy, i.e., all the higher weights, of a given class of codes. For example, in a significant piece of work, Heijnen and Pelikaan \cite{HP} completely determined the higher weights of Reed-Muller codes $\RM(\nu,\delta)$. 
%The corresponding question for 
In the case of projective Reed-Muller codes,  the minimum distance was determined by Lachaud \cite{L2} and independently by S{\o}rrensen \cite{So}. In fact, Lachaud derives it as a consequence of an affirmative answer given by Serre \cite{Se} to a question of Tsfasman concerning the maximum number of $\Fq$-rational points on a projective hypersurface of a given degree. The second higher weight  was determined by Boguslavsky \cite{Bog}, while 
the determination of $d_r\left(\PRM(\nu,\delta)\right)$ is still open for $r>2$. In the case of Grassmann codes, the  $r^{{\rm th}}$ higher weight is known for the first few and the last few values of $r$, thanks to Nogin \cite{N} (see also \cite{GL}) and Hansen, Johnsen and Ranestad \cite{HJR} (see also \cite{GPP}). More precisely, for $r=0,1, \dots \mu$, where $\mu:=1+\max\{\ell,m-\ell\}$, we have
%\begin{equation} \label{dualweights}
$$
d_r\left( C(\ell ,m)\right) = q^{\d} + q^{\d - 1} +\dots +   q^{\d - r +1}  \quad \text{and} \quad 
d_{k-r}\left( C(\ell ,m)\right) = n -( 1+ q + \cdots + q^{r -1}), 
$$
%{\rm for} \quad 0\le r \le \mu. \end{equation}
where $n$ denotes the length of $C(\ell,m)$ or in other words, the number of $\Fq$-rational points of $G_{\ell,m}$, and it is given by the Gaussian binomial coefficient ${{m}\brack{\ell}}_q $. 
In case $\ell =2$, we know a little more (cf. \cite{GPP}), but the general case is still open. 

We consider in this paper the problem of determining the higher weights of affine Grassmann codes and their duals. Our main result is %that for the affine Grassmann code 
an explicit formula for $d_r\left(C^{\A}(l,m)\right)$ for the first few and the last few values of $r$, or more precisely, for $0\le r \le \mu'$ % \ell'-\ell + 1$ 
and for $k- \mu  \le r \le k$, where
\begin{equation}
\label{mumu'}
\mu' = 1 + \max\{\ell, \, \ell'-\ell \} \quad \text{ and } \quad \mu := 1 + \max\{\ell, \, \ell'\} = \ell'+1.
\end{equation}
In the case of the result for the first $\mu'$ higher weights, we have to make an additional mild assumption that $\ell < \ell'$. The result for the last $\mu$ higher weights can be deduced from the corresponding results for Grassmann codes using a geometric approach. However, we give here self-contained proofs in the spirit of \cite{BGH1,BGH2} and this has the advantage that analogous results are also obtained for affine Grassmann codes of arbitrary level introduced in \cite{BGH2}. As for the duals, we can in fact go much farther, and determine many more higher weights of the duals of affine Grassmann codes except that the result we give here is best described recursively. As a %n easy 
corollary, we obtain a new and simpler proof of \cite[Theorem 17]{BGH2}, which states that if $\ell'>1$, then  the minimum 
distance of $C^{\A}(l,m;h)^{\perp}$ is $3$ or $4$ according as $q>2$ or $q=2$. 
% just make a small beginning beyond the determination the result in \cite{BGH2} about the %minimum distance of $C^{\A}(l,m)^{\perp}$.  Namely, we explicitly determine 
%$d_2\left(C^{\A}(l,m)^{\perp}\right)$. 
The geometric approach and an alternative proof of the result about the last $\mu$ higher weights is also outlined in an appendix for the convenience of the reader. 
%
%This paper is organized as follows. 

\section{Initial Higher Weights}
\label{sec2}

{For} any $q$-ary linear code $C$ of length $n$, and any $D\subseteq C$, we let 
$$
\supp (D):=\left\{i\in\{1, \dots , n\} : c_i\ne 0 \text{ for some } c\in D\right\} \quad \text{and} \quad \w (D) = |\supp (D)|
$$
denote, respectively, the \emph{support} and the \emph{support weight} of $D$. For a codeword 
$c = (c_1, \dots , c_n)\in C$, we write $\w (c) = \w(\{c\})$ and note that this is simply the \emph{Hamming weight} of $c$.  
 
Fix, throughout this paper, positive integers $h, \ell, \ell'$ with $h \le \ell \le \ell'$ and an $l \times l'$ matrix $X= \left(X_{ij}\right)$ whose entries are algebraically independent indeterminates over $\Fq$. 
Let $\Fq[X]$ denote the ring of polynomials in the $\ell\ell'$ variables $X_{ij}$'s with coefficients in $\Fq$. As in \cite{BGH2},  
%for $0\le i\le \ell$, we let $\Delta_i(\ell,m)$ denote the subset of $\F[X]$ consisting of all $i \times i$ minors of $X$, and 
we let $\Delta(\ell,m;h)$  denote the  set of all minors of $X$ of degree $\le h$. Note that 
$\Delta(\ell,m;h)$  is a subset of $\Fq[X]$ that contains the constant polynomial $1$, which corresponds to the $0\times 0$ minor of $X$. Further let %$\flmh$ denote the 
$$
\flmh := \mbox{the $\Fq$-linear  subspace of $\Fq[X]$ generated by } \Delta(\ell,m;h).
$$
Note that the space $\flm$ defined in the Introduction contains $\flmh$ and the equality holds when $h=\ell$.  
%moreover,  $\flm=\flml$.
The \emph{affine Grassmann code of level $h$}, denoted $\Ch$,  is defined to be the image of $\flmh$ under the evaluation map $\Ev$ given by \eqref{EvonFlm}. Evidently, $\Cl =\C$ and $\Ch$ is a subcode of %the Reed-Muller code 
$\RM(h,\delta)$. % and $\Cl=\C$. 
%
%Now let $\A^{\delta}$ denote the $\delta$-dimensional affine space of all $ \ell \times \ell'$ matrices with entries in $\Fq$, and, as in the Introduction, let  $P_1, \dots , P_n$ is 
%an ordered listing %a listing 
%of $\A^{\delta}$. 
%
%For any $q$-ary linear code $C$ of length $n$ and any $c = (c_1, \dots , c_n)\in C$, we denote %by $\w(c)$ the Hamming weight of $c$, i.e., $\w(c) = |\{i\in\{1, \dots , n\} : c_i\ne 0\}|$. We now  begin 
Now here is a slightly refined version of a basic result proved in \cite{BGH2}.
%
%Let us begin with a slightly refined version of \cite[Theorem 5]{BGH2}. 
%the result concerning the minimum distance of $\Ch$. 

\begin{proposition}
\label{prop:mindist}
The minimum distance $d(\ell,m;h)$ of $\Ch$ is
\begin{equation}\label{eq:dlmr2}
d(\ell,m;h)  = q^{\delta}\prod_{i=1}^{h}\left(1 - \frac{1}{q^i}\right) = q^{\delta - h^2} \left|\GL_{h}(\Fq)\right| .
\end{equation}
Moreover, if $\M$ is any $h\times h$ minor of $X$, then 
$\w\left(\Ev\left(\M\right)\right) = d(\ell,m;h)$. 
%$ is a minimum-weight codeword of $\Ch$.
\end{proposition}

\begin{proof}
The first equality \eqref{eq:dlmr2} is proved in \cite[Theorem 5]{BGH2}, while the second is easily deduced. % immediate. 
Also it is shown that in \cite[Theorem 5]{BGH2} if $\Lh = \det\left(X_{ij}\right)_{1\le i,j \le h}$ %$\Lh(X)= =1, \dots , h}$ 
is the $h^{\rm th}$ leading principal minor of $X$, then 
$\w\left(\Ev\left(\Lh\right)\right) = d(\ell,m;h)$. Now if $\M$ is any $h\times h$ minor of $X$, then there are %$p_1, \dots , p_h\in \{1,\dots , \ell\}$ and $q_1, \dots , q_h\in \{1,\dots , \ell'\}$ 
positive integers $p_1, \dots , p_h, \, q_1, \dots , q_h$
with $p_1<\cdots < p_h\le \ell$ and $q_1<\cdots < q_h\le \ell'$ such that $\M = \det (X_{p_iq_j})_{1\le i,j\le h}$.  Let $\sigma\in S_{\ell}$ be a permutation such that $\sigma(i)=p_i$ for $1\le i \le \ell$ and $P \in \GL_{\ell}(\Fq)$ be the permutation matrix corresponding to $\sigma$ so that for $1\le i,j \le \ell$,  the $(i,j)^{\rm th}$ entry of $P$ is $1$ is $j=\sigma (i)$ and $0$ otherwise. Likewise, let $\tau\in S_{\ell'}$ be such that $\tau(i) =q_i$ for $1\le i\le \ell'$ and $Q\in \GL_{\ell'}(\Fq)$ be the permutation matrix corresponding to $\tau$. Then 
it is easily seen that 
%$\M = \Lh(PXQ^{-1})$, i.e., 
$\M$ is the the $h^{\rm th}$ leading principal minor of $PXQ^{-1}$.
Moreover, we know from \cite[\S IV]{BGH2} that $X\mapsto PXQ^{-1}$ induces a permutation  automorphism of $\Ch$. It follows that $\w\left(\Ev\left(\M\right)\right) = \w\left(\Ev\left(\Lh\right)\right) = d(\ell,m;h)$.  
\end{proof}

The following general observation about the support weights of linear codes will be useful in the sequel. 

\begin{lemma}
\label{union}
Let $C$ be an $[n,k]_q$-code and for $i=1, \dots , n$, let $\pi_i: C \to \Fq$ denote the $i^{\rm th}$ projection map defined by $\pi_i(c_1, \dots, c_n) = c_i$.
Also let $D$ be a subcode of $C$ %of dimension $r$ 
and $\{y_1, \dots, y_r\}$ be a generating set of $D$. Then 
$$
\supp (D) = \bigcup_{j=1}^n A_j \quad \text{where for } 1\le j\le n, \quad 
A_j :=\left\{i\in\{1, \dots , n\} : \pi_i (y_j) \neq 0 \right\}.
$$
\end{lemma}

\begin{proof}
%It is clear that 
Clearly, $\cup_{j=1}^n A_j \subseteq \supp(D)$. On the other hand, suppose $i\in \{1, \dots  , n\}$ is such that $i \notin \cup_{j=1}^n A_j$. Then  $\pi_i (y_j) = 0$ for all $j=1, \dots, r$.  
Now for any $x \in D$,  we can write $x = \sum_{j=1}^{r} c_j y_j$ for some $c_1, \dots , c_r\in \Fq$; % and 
hence $\pi_i(x) = \sum_{j=1}^{r} c_j \pi_i(y_j) = 0$.  Thus $i \notin \supp(D)$. 
This shows that $\supp(D) \subseteq \cup_{j=1}^n A_j $. 
\end{proof}

The next two lemmas extend Proposition \ref{prop:mindist} and show that for a judicious choice of a family $\{\M_1, \dots , \M_r\}$ of minors %$\M_1, \dots , \M_r$ 
of $X$, the support weight of the product of any nonempty subfamily 
is given by a formula analogous to 
\eqref{eq:dlmr2}. %As such the two lemmas may be viewed as an extension of Proposition \ref{prop:mindist}. 

\begin{lemma}
\label{closeminors}
Let $r$ be a positive integer such that $r\le \ell' - h +1$ and  
let $Y$ be any $h\times (h+r-1)$ submatrix of $X$.  
Also for $j=1, \dots , r$, let $\M_j$ denote the $h\times h$ minor of $Y$ (and hence of $X$) %formed by 
corresponding to the first $h-1$ columns of $Y$ together with the $(h+j-1)^{\rm th}$ column of $Y$,  
and let $A_j =\{P \in \mathbb{A}^{\delta} : \mathcal{M}_j(P) \neq 0 \}$. Then for any positive integer $s$ with $s\le r$ and any $j_1, \dots , j_s\in \{1,\dots , r\}$ with $j_1 < \cdots < j_s$, 
\begin{equation}
\label{Aint}
%\left|\bigcap_{j=1}^s A_j \right| 
\left| A_{j_1} \cap \cdots \cap A_{j_s} \right|= d(\ell,m;h) \left(1 - \frac 1q\right)^{s-1}.
\end{equation}
\end{lemma}
 
\begin{proof}
Given any $\ell\times \ell'$ matrix $P\in  \mathbb{A}^{\delta} $ with entries in $\Fq$, let $Q$ denote the $h\times (h+r-1)$ submatrix of $P$ formed in exactly the same way as $Y$, and let 
$Q_1, \dots Q_{h+r-1}$ denote the column vectors of $Q$. %Now 
For any positive integer $s$ with $s\le r$ and any $j_1, \dots , j_s\in \{1,\dots , r\}$ with $j_1 < \cdots < j_s$, the 
condition $P\in A_{j_1} \cap \cdots \cap A_{j_s}$ is equivalent to the condition that the column vectors  $Q_1, \dots , Q_{h-1}, Q_{h+j-1}$ in $\Fq^h$ are linearly independent for each $j\in \{j_1, \dots, j_s\}$. This will hold when the submatrix of $Q$ formed by its first $h-1$ columns 
is chosen %Thus $Q_1, \dots , Q_{h-1}$ can be chosen 
in exactly $(q^h-1)(q^h-q)\cdots (q^h-q^{h-2})$ ways, while each of $Q_{h+j_1-1}, \dots , Q_{h+j_s-1}$ are chosen in $(q^h-q^{h-1})$ ways. The remaining $r-s$ columns of $Q$ may be chosen arbitrarily in $q^{h(r-s)}$ ways. Since $P$ has $\ell\ell' - h(h+r-1)$, i.e.,  $\delta - h^2 - h(r -1)$, entries outside $Q$, it follows that 
\begin{eqnarray*}
\left| A_{j_1} \cap \cdots \cap A_{j_s} \right| & = &(q^h-1)(q^h-q)\cdots (q^h-q^{h-2}) (q^h-q^{h-1})^s q^{\delta - h^2 - h(s -1)} \\
&=&  d(\ell,m;h) \left(1 - \frac 1q\right)^{s-1},
%\left|\GL_{h}(\Fq)\right| q^{\delta - h^2} \left(\frac {q^h-q^{h-1}}{q^{h}}\right)^{s-1},
\end{eqnarray*}
%which gives the desired equality, 
%thanks to
where the last equality follows from \eqref{eq:dlmr2}.
%  d(\ell,m;h) \left(1 - \frac 1q\right)^{s-1}.
\end{proof}

\begin{lemma}
\label{closeminors2}
Assume that $h < \ell'$. 
Let $r$ be a positive integer such that $r\le h +1$ % \le \ell'$ 
and  let $Y$ be any $h\times (h+1)$ submatrix of $X$.  
Also for $j=1, \dots , r$, let $\M_j$ denote the 
determinant of the $h\times h$ submatrix of $Y$ formed by all except the  $(h-r+j+1)^{\rm th}$ column of $Y$,  
and let $A_j =\{P \in \mathbb{A}^{\delta} : \mathcal{M}_j(P) \neq 0 \}$. Then \eqref{Aint} holds for any positive integer $s$ with $s\le r$ and any $j_1, \dots , j_s\in \{1,\dots , r\}$ with $j_1 < \cdots < j_s$.
\end{lemma}
 
\begin{proof}
%As in the proof of Lemma \ref{closeminors}, 
Given $P\in \mathbb{A}^{\delta}$, let $Q$ be the $h\times (h+1)$ submatrix of $P$ corresponding to $Y$, and let $Q_1, \dots, Q_{h+1}$ denote its column vectors. Fix 
any positive integer $s$ with $s\le r$ and  $j_1, \dots , j_s\in \{1,\dots , r\}$ with $j_1 < \cdots < j_s$. Now $\M_{j_1}(P)\ne 0$ implies that $Q$ has rank $h$ and in particular, $Q_{h-r+j_1+1}$ is a $\Fq$-linear combination of the remaining $h$ column vectors of $Q$.  Moreover, for $2\le t \le s$, if $\M_{j_t}(P)\ne 0$, %for each $t=2, \dots , s$, 
then the coefficients of $Q_{h-r+j_t+1}$ in this $\Fq$-linear combination must 
be nonzero. % for each $t=2, \dots , s$.  
Conversely, if all except 
the $(h-r+j_1+1)^{\rm th}$ column of $Q$ are linearly independent (and these columns 
can thus be chosen in $\left|\GL_{h}(\Fq)\right|$ ways), while 
$Q_{h-r+j_1+1}$ is a $\Fq$-linear combination of the remaining $h$ column vectors of $Q$ with a nonzero coefficient for the $s-1$ columns $Q_{h-r+j_2+1}, \dots , Q_{h-r+j_s+1}$, then 
$\M_{j_t}(P)\ne 0$ for each $t=1, \dots , s$. The $h$ coefficients in this $\Fq$-linear combination can thus be chosen in $q^{h-s+1}(q-1)^{s-1}$ ways. Since $P$ has $\delta - h(h+1)$ entries outside of $Q$, it follows that 
$$
\left| A_{j_1} \cap \cdots \cap A_{j_s} \right|  =  \left|\GL_{h}(\Fq)\right|  q^{h-s+1} 
(q-1)^{s-1} q^{\delta - h(h+1)} =  d(\ell,m;h) \left(1 - \frac 1q\right)^{s-1},
$$
%%\left|\GL_{h}(\Fq)\right| q^{\delta - h^2} \left(\frac {q^h-q^{h-1}}{q^{h}}\right)^{s-1},
%\begin{eqnarray*}
%\left| A_{j_1} \cap \cdots \cap A_{j_s} \right| & = & \left|\GL_{h}(\Fq)\right|  q^{h-s+1}(q-1)^{s-1} q^{\delta - h(h+1)} \\
%&=&  d(\ell,m;h) \left(1 - \frac 1q\right)^{s-1},
%%\left|\GL_{h}(\Fq)\right| q^{\delta - h^2} \left(\frac {q^h-q^{h-1}}{q^{h}}\right)^{s-1},
%\end{eqnarray*}
where the last equality follows once again from \eqref{eq:dlmr2}.
%which gives the desired equality,
\end{proof}

\begin{theorem}
\label{mainthm1}
Let $r$ be a positive integer such that $r\le  \max\{\ell' - h, h\} +1$. Assume that $h<\ell'$ in case $\max\{\ell' - h, h\}  =h$, i.e., $\ell' \le 2h$. 
Then the $r^{{\rm th}}$ {higher weight}   $d_r(\ell,m;h)$ of $\Ch$ is
\begin{equation}\label{eq:dlmr}
d_r(\ell,m;h)  = q^{\delta-r+1}\frac{(q^r-1)}{(q-1)}\prod_{i=1}^{h}\left(1 - \frac{1}{q^i}\right) = q^{\delta - h^2-r+1} (q^r-1)\frac{\left|\GL_{h}(\Fq)\right|}{q-1} .
\end{equation}
Moreover, the $r^{{\rm th}}$ {higher weight}  of $\Ch$ attains the Griesmer-Wei bound. 
\end{theorem}
 
\begin{proof}
The hypotheses on $r$ and $h$ together with Lemmas \ref{closeminors} and \ref{closeminors2} ensure that there exist  minors 
$\M_1, \dots , \M_r \in \Delta(\ell,m;h)$ with supports $A_1, \dots , A_r$ respectively, such that 
%if $A_j =\{P \in \mathbb{A}^{\delta} : \mathcal{M}_j(P) \neq 0 \}$ for $1\le j\le r$. Then 
\eqref{Aint} holds for any positive integer $s$ with $s\le r$ and any 
\mbox{$j_1, \dots , j_s\in \{1,\dots , r\}$} with $j_1 < \cdots < j_s$. Consequently, 
\begin{eqnarray*}
\left|\bigcup_{j=1}^r A_j \right|  & = & \sum_{s=1}^r (-1)^{s-1} \sum_{1\le j_1 < \cdots < j_s \le r} \left| A_{j_1} \cap \cdots \cap A_{j_s} \right|  \\ %\nonumber \\
& = & \sum_{s=1}^r (-1)^{s-1} {{r}\choose{s}} d(\ell,m;h) \left(1 - \frac 1q\right)^{s-1} \\
& = & \frac{d(\ell,m;h) }{1 - q^{-1}} \left( 1 - \sum_{s=0}^r (-1)^{s} {{r}\choose{s}}  \left(1 -q^{-1}\right)^{s} \right) 
\\
& = & \frac{d(\ell,m;h) }{1 - q^{-1}} \left( 1 -  \left[ 1 - \left(1 -q^{-1}\right)\right]^{r} \right). 
\end{eqnarray*} 
%\\
%& = & d(\ell,m;h) \frac{(q^r-1)}{q^{r-1}(q-1)}.
%\end{eqnarray*}
Hence, in view of Lemma \ref{union} and Proposition \ref{prop:mindist}, we see that if $D_r$ is the subspace of $\flmh$ spanned by $\M_1, \dots , \M_r$, then 
\begin{equation}
\label{eq:wtDr}
 \w(D_r) = d(\ell,m;h) \frac{(q^r-1)}{q^{r-1}(q-1)} = q^{\delta-r+1}\frac{(q^r-1)}{(q-1)}\prod_{i=1}^{h}\left(1 - \frac{1}{q^i}\right).
% = q^{\delta - h^2-r+1} (q^r-1)\frac{\left|\GL_{h}(\Fq)\right|}{q-1} .
\end{equation}
Moreover, by \cite[Lemma 3]{BGH1}, we see that $D_r$ is of dimension $r$. Thus,
\begin{equation}
\label{ineq:dlmr}
d_r(\ell,m;h)  := d_r\left(\Ch\right) \le d(\ell,m;h) \frac{(q^r-1)}{q^{r-1}(q-1)} = \sum_{i=0}^{r-1} \frac{ d(\ell,m;h)}{q^{i}}.
\end{equation}
On the other hand, the Griesmer-Wei bound (cf. \cite{W}) %together with 
and Proposition \ref{prop:mindist} yields
\begin{equation}
\label{ineq:GW}
d_r\left(\Ch\right) \ge \sum_{i=0}^{r-1} \left\lceil \frac{ d(\ell,m;h)}{q^{i}} \right\rceil  \ge \sum_{i=0}^{r-1} \frac{ d(\ell,m;h)}{q^{i}}.
\end{equation}
Using \eqref{eq:wtDr}, \eqref{ineq:dlmr} and \eqref{ineq:GW},  we obtain the desired result.
\end{proof}

\begin{remark}
The only case in which the above theorem does not give any higher weights of $\Ch$ beyond the minimum distance is when $h=\ell = \ell'$. 
We believe that in this case even the second higher weight %$d_2(\Ch)$ 
does not meet the Griesmer-Wei bound. In fact, it seems plausible that for $C^{\A}(\ell, 2\ell) = 
C^{\A}(\ell, 2\ell; \ell)$, 
$$
d_2\left(C^{\A}(\ell, 2\ell) \right) = q^{\ell}
\left( 1 + \frac{1}{q} - \frac{1}{q^{\ell}} \right)
\prod_{i=0}^{\ell-1} \left( q^{\ell} - q^i\right) 
= d(\ell, 2\ell; \ell) \left( 1 + \frac{q^{\ell-1} }{q^{\ell}-1} \right) .
$$
Note that  the expression on the right is strictly greater than $d(\ell, 2\ell; \ell) \left( 1 + q^{-1}\right)$. %\frac 1q \right) $. 
We remark also that the expression on the right is an upper bound for $d_2\left(C^{\A}(\ell, 2\ell) \right)$. This can be seen, for example, by considering the $2$-dimensional subspace $D$ of ${\EuScript{F}(\ell,2\ell)}$ spanned by the $\ell^{\rm th}$ leading principal minor $\M_1 = \det(X)$ and the $(\ell-1)^{\rm th}$ leading principal minor $\M_2$ of $X$, and using Lemma \ref{union} to show that 
\begin{eqnarray*}
|\supp (D) | & = &  |A_1| + |A_2| - |A_1\cap A_2| \\
&=& |\GL_{\ell}(\Fq)| + |\GL_{\ell-1}(\Fq)|q^{2\ell-1} - \GL_{\ell-1}(\Fq)|q^{\ell-1}(q^{\ell}- q^{\ell-1}) \\
&=& d(\ell, 2\ell; \ell) \left( 1 + \frac{q^{\ell-1} }{q^{\ell}-1} \right) ,
\end{eqnarray*}
where, as before, $A_j =\{P \in \mathbb{A}^{\delta} : \mathcal{M}_j(P) \neq 0 \}$ for $j=1,2$. 
\end{remark}

\section{Terminal Higher Weights}
\label{sec:terminal}

As in the case of Grassmann codes, determining some of the terminal higher weights is 
simpler than determining some of the initial higher weights. 
In fact, we go a little farther than what we could do with the initial higher weights. Thus, as opposed to finding explicitly the first $\mu'$ higher weights of say 
$\C$, where $\mu' = \max\{\ell, \ell'-\ell\}+1$, we are able find explicitly the last $\ell' + 1$ higher weights of not just $\C$, but any % and more generally, of  
$\Ch$. 

The first step is a simple observation that holds, in fact, for any functional code defined by means of an evaluation map on a space of (polynomial) functions. However, we will just restrict to the case of affine Grassmann codes of a given level. 

\begin{lemma}
\label{ZW}
Let $D$ be a subcode of $\Ch$ and $t$ a positive integer such that there exist $t$ linearly independent polynomials $g_1, \dots , g_t\in \Fq[X]$ with the property that $\deg g_i \le 1 $ and $\Ev (g_i) \in D$ for all $i=1, \dots , t$. Then $\w(D) \ge q^{\delta} - q^{\delta - t}$. 
\end{lemma}
 
\begin{proof}
As noted in \cite[\S II]{BGH1}, the evaluation map $\Ev$ given by \eqref{EvonFlm} is injective. Hence $D$ is in bijection with %Consider 
$W:=\Ev^{-1}(D)=\{f\in \flmh: \Ev(f)\in D\}$. Moreover, if we let 
$Z(W) :=\{P\in \A^{\d} : f(P)=0 \text{ for all } P\in W\}$ denote the corresponding 
affine variety. then it is clear that 
$$
\w(D) = |\supp (D )| = |\A^{\d} \setminus Z(W) |  = q^{\d} - |Z(W)|.
$$
Now $g_1, \dots, g_t\in W$ and the number of common zeros in $\A^{\d}$ of $g_1, \dots , g_t$ corresponds to the number, say $N$, of solutions of a system of $t$ linearly independent nonhomogeneous linear equations in $\d$ variables with coefficients in $\Fq$. Hence $N=0$ or $N= q^{\delta - t}$ according as the system is inconsistent or consistent. Consequently, $|Z(W)|\le N \le q^{\delta - t}$
and so $\w(D) \ge q^{\delta} - q^{\delta - t}$. 
\end{proof}

Let $k_h$ denote the dimension of $\Ch$. % will be denoted by $k_h$. 
We know from \cite[Prop. 2]{BGH2} that
\begin{equation}
\label{kh}
 k_h = \sum_{i=0}^h {{\ell}\choose{i}} \binom{\ell'}{i}  \quad \text{ and } \quad k_{\ell} = \binom{m}{\ell}.
\end{equation}
In the remainder of this paper, we fix an ordering $\N_1, \dots , \N_{k_h}$ on $\Delta(\ell,m;h)$ such that the $0\times 0$ minor appears at the end, preceded by the $1\times 1$ minors arranged lexicographically. % order. 
More %specifically
precisely, we require $\N_{k_h} = 1$ and $\N_{k_h-(i-1)\ell'-j} = X_{ij}$ for $i=1, \dots, \ell$ and $j=1, \dots , \ell'$. For instance, $\N_{k_h-1} = X_{11}$ and $\N_{k_h-\d} = X_{\ell \ell'}$. 
Note that $\N_1, \dots , \N_{k_h}$ gives an ordered $\Fq$-basis of $\flmh$, thanks to  \cite[Prop. 2]{BGH2}. 

\begin{lemma}
\label{BoundOndr}
$d_{k_h-r}(\Ch) \ge q^{\delta} - q^{r- 1}$ for %any positive integer $r$ with $r\le \d +1$. 
all $r=1, \dots , \d +1$. 
\end{lemma}

\begin{proof}
{Fix} a positive integer $r$ with $r\le \d +1$ and let $s := k_h-r$. Let $D$ be any %$(k_h-r)$
$s$-dimensional subcode of $\Ch$. Then $W:=\Ev^{-1}(D)$ is a linear subspace of $\flmh$ with $\dim W =s$. Suppose $f_1, \dots , f_s$ give an ordered $\Fq$-basis of $W$. Then there is a unique $s\times k_h $ matrix $\Lambda = \left( \lambda_{ij}\right)$ of rank $s$ %and 
with entries in $\Fq$ such that $\underline{\mathbf{f}} = \Lambda \underline{\mathbf{\N}}$, where 
$\underline{\mathbf{f}} := [f_1, \dots , f_s]^T$ and 
$\underline{\mathbf{\N}} := \left[\N_1, \dots , \N_{k_h}\right]^T$ 
denote the column vectors (over $\Fq[X])$ corresponding to the  abovementioned ordered 
$\Fq$-bases of $W$ and $\flmh$ respectively. 
Now let $\Lambda^* = \left( \lambda_{ij}^*\right)$ be the reduced row-echelon form of $\Lambda$ and let $p_1, \dots , p_s$ with $1\le p_1< \cdots < p_s\le k_h$ be the column indices in which the pivots occur, so that $\lambda_{ip_i}^*=1$ and $\lambda_{ij}^* =0$ if $j<p_i$ and also $\lambda_{kp_i}^* =0$ if $k\ne i$. It is clear that 
$p_s\ge s \ge k_h - \delta -1$.  Hence if we let $t:=\delta+1-r$, then 
in each of the last $t$ rows of $\Lambda^*$, the first $k_h - \delta -1$ entries are zero.  
Consequently,  the last $t$ %\delta+1-r$ 
rows of the product 
$ \Lambda^*\underline{\mathbf{\N}}$ are $\Fq$-linear combinations of the last $\delta+1$ minors among $\N_1, \dots , \N_{k_h}$,  and these give rise to  linearly independent polynomials $g_1, \dots , g_t \in \Fq[X]$ of degree $\le 1$.  Moreover, $\Lambda^*$ is obtained from $\Lambda$ by a finite sequence of elementary row operations, and hence there is a nonsingular $s\times s$ matrix $P$ with entries in $\Fq$ such that $\Lambda^* = P\Lambda$. In particular, $g_1, \dots , g_t$ correspond to the last $t$ rows of the product $P\underline{\mathbf{f}}$, and hence they are in $W$. Now Lemma \ref{ZW} implies that 
$\w(D) \ge q^{\delta} - q^{\delta - t} = q^{\delta} - q^{r- 1}$. 
Since $D$ was an arbitrary $s$-dimensional subscode of $\Ch$, we obtain the desired result. 
\end{proof}

It can be shown that the lower bound in Lemma \ref{BoundOndr} is attained when $r\le \ell'+1$, and this leads to the following result about the terminal higher weights. 

\begin{theorem}
\label{ValueOfdr}
$d_{k_h}(\Ch) = q^{\delta} $ and 
$d_{k_h-r}(\Ch) = q^{\delta} - q^{r- 1}$ for any positive integer $r$ with $r\le \ell' +1$. 
%all $r=1, \dots , \ell' +1$. 
\end{theorem}

\begin{proof}
Since $1\in \flmh$, we see that $d_{k_h}(\Ch) = q^{\delta} $ or in other words, the code $\Ch$ is nondegenerate. 
Fix a positive integer %$r$ with 
$r\le \d +1$ and let $s := k_h-r$. Consider the 
linear subspace $W$ of $\flmh$ spanned by $\{ \N_1, \dots , \N_s\}$, i.e., by all 
%except the last $r$ 
the minors in $\Delta (\ell, m;h)$, except $1$ and $X_{1j}$ for $j=1,  \dots , r-1$. 
%$1, X_{11}, \dots , X_{1r-1}$. 
Observe that the corresponding affine variety $Z(W)$ consists precisely of the $\ell \times \ell'$ matrices  
$P=(P_{ij})\in \A^{\d}$ satisfying  $P_{ij}=0$ for all $i=1, \dots , \ell$ and 
$j=1, \dots , \ell'$, except when $(i, j)=(1, 1), \dots , (1, r-1)$. % and $j=1, \dots , r-1$. 
Consequently, $|Z(W)| = q^{r-1}$ and thus if we let $D:=\Ev(W)$ be the $s$-dimensional subcode of $\Ch$ corresponding to $W$, then $\w(D) =  q^{\delta} - q^{r- 1}$. This together with Lemma \ref{BoundOndr} yields the desired result. 
\end{proof}

\section{Higher Weights of Duals of Affine Grassmann Codes}
\label{sec:duals}

%%%%% Older Beginning%%%%%%%%%%%%%%
%The dual of $\Ch$ has been explicitly described in \cite{BGH2}. We recall it below after 
%setting up some notation and terminology that will be useful in the sequel. 
%
%Let $\Mon$ denote the subset of $\Fq[X]$ consisting of all monomials in the $\ell\ell'$ variables $X_{ij}$. For example, 
%$$
%\full =  \prod_{i=1}^{\ell} \prod_{j=1}^{\ell'} X_{ij}^{q-1}
%$$ is an element of $\Mon$ that we refer to, following \cite{BGH2}, as the \emph{full product}. %We may define a monomial in $\Mon$ to be \emph{reduced} if it divides $\full$, and we define %$\RMon$ to be the set of all reduced monomials in $\Fq[X]$. For a minor 
%$\M\in \Delta(\ell, m;h)$, we denote by $\Term (\M)$ the set of all monomials appearing in %$\M$.  Define
%$$
%\FM   := \left\{\frac{\full}{t} \, : \, t\in \Term({\M}) \text{ for some } {\M}\in %\Delta(\ell,m;h)\right\}.
%$$
%%%%%%%%%% End of Older Beginning %%%%%%%%%%%%%

{For} determining the higher weights of duals of affine Grassmann codes, we use a simple, but powerful, method based on the following key result of Wei \cite{W}.

\begin{proposition} 
\label{Wei}
Let $C$ be a $[n,k]_q$-code. Then 
\begin{enumerate}
\item[{\rm (i)}] \emph{(Monotonicity)} If $k>0$, then $1\le d_1(C)< d_2(C) < \dots < d_k(C)\le n$.
\item[{\rm (ii)}] \emph{(Duality)} The higher weights of $C$ and its dual are related by
$$
\left\{ d_s(C^\perp): s = 1, \dots , n-k\right\} = 
\left\{  1, \dots , n\right\} \setminus
\left\{ n+1 - d_r(C): r = 1, \dots , k\right\}.
$$
\end{enumerate}
\end{proposition}

It is convenient and computationally effective to rephrase the above result for nondegenerate linear codes of positive dimension as follows. We will in fact give two equivalent formulations, the first of which is better suited for the terminal weights while the second is better suited for the initial weights. 

\begin{corollary} 
\label{CorWei}
%Let $C$ be a $[n,k]_q$-code. Assume that $k>0$ and $C$ is nondegenerate. 
Let $C$ be a nondegenerate $[n,k]_q$-code with $k>0$.
Let $d_0 :=0$ and for $1\le r\le k$, let $d_r$ %= d_r(C)$.  
denote the $r^{th}$ higher weight of $C$. 
Also let 
\begin{equation}
\label{eandf}
e_j := d_j - j \quad \text{ and } \quad f_j := n - j - d_{k-j} \quad \text{ for } \ 0\le j \le k.
\end{equation}
Then the $e$-sequence and the $f$-sequence partition $\{0,1,\dots , n-k\}$; 
%the integral interval between $0$ and $n-k$, or 
more precisely, %we have
\begin{equation}
\label{eandfIneq}
0 = e_0 \le e_1 \le \cdots \le e_k = n-k  \quad \text{ and } \quad 0=f_0 \le f_1 \le \cdots \le f_k = n-k .
\end{equation}
Moreover,  for $0\le s < n-k$, the last $s^{th}$ higher weight %$d_s^{\perp} := d_s (C^\perp)$ 
of the dual of $C$ is given by 
\begin{equation}
\label{DualTerminal}
d_{n-k - s}\left(C^\perp\right) = n-s - j \quad \text{ if } j % \in \{0, 1, \dots , n-k\} 
\text{ is the unique integer $< k $  with } % $\le n-k$ with } %such that } 
e_j \le s < e_{j+1}. % \text { for some $j$ with } 0\le j \le k.
\end{equation}
%where %we set 
%$e_{k+1} = \infty$, by convention. 
Equivalently,  for $0 <  s \le n-k$, the $s^{th}$ higher weight of the dual of $C$ % $C^\perp$ 
is given by 
\begin{equation}
\label{DualInitial}
d_{s}\left(C^\perp\right) = s+ j + 1 \quad
\text{ if } j  %\in \{0, 1, \dots , k\} % \text{ is the unique integer $\le k $ with } 
\text{ is the unique integer $< k $ with } 
 f_j < s \le  f_{j+1}. % \text { for some $j$ with } 0\le j \le k,
\end{equation}
%where %we set 
%$f_{k+1} = \infty$, also by convention.
%\end{enumerate}
\end{corollary}

\begin{proof} Note that $d_k = n$, since $C$ is nondegenerate. With this in view, part (i) of Proposition \ref{Wei} implies \eqref{eandfIneq}. Next, parts (i) and (ii) of Proposition \ref{Wei} together with \eqref{eandfIneq} readily imply  \eqref{DualTerminal} and  \eqref{DualInitial} .
\end{proof}

\begin{remark}
The above Corollary shows that the higher weights of the dual of a nondegenerate linear code $C$ of positive dimension $k$ take consecutive values in strings of length $d_{r+1} - d_r -1$ for $0\le r < k$, where $d_r$ denotes the $r^{th}$ higher weight of $C$.  Evidently, this phenomenon is prevalent if there are large gaps among the consecutive higher weights of $C$. In fact, a duality of sorts seems to prevail here: more the number of consecutive strings among the higher weights of a code, the less there are among the higher weights of its dual, and vice-versa. In this connection, it may useful to note the following result of Tsfasman and Vl\u{a}du\c{t} \cite[Cor. 3.5]{TV2}, which states that 
%if $d_r$ denotes the $r^{th}$ higher weight of a $[n,k]_q$-code $C$, then 
for $1\le r \le s \le k$, 
$$
d_s \ge d_r + \sum_{i=1}^{s-r} \left\lceil \frac{(q-1)d_r}{(q^r - 1) q} \right \rceil
\quad \text{ and in particular, } \quad d_{r+1} - d_r \ge \left\lceil \frac{(q-1)d_r}{(q^r - 1) q} \right \rceil. 
$$
Another special case of Corollary \ref{CorWei} worth noting is that $C^\perp$ is nondegenerate if and only if $d_1(C) > 1$. 
\end{remark}

We now turn to duals of affine Grassmann codes. 
Recall that we have fixed positive integers $h, \ell, \ell'$ with $h\le \ell \le \ell'$ and that the length of the corresponding affine Grassmann code $\Ch$ of level $h$ is given by $n:=q^{\d}$ and the dimension $k_h$ is given by  \eqref{kh}. To avoid trivialities we will further assume that 
$\ell'> 1$. 
%In what follows, we will ignore the trivial cases when $h=0$ or when $\ell'=1$ (which implies 
%$\ell =1$ and $h=0$ or $1$). 
Indeed, it is easy to describe what the affine Grassmann code $\Ch$ and its dual is in the trivial case  $\ell'=1$ (or another trivial case $h=0$ that we have ignored from the beginning)  and in fact, this has been done in the paragraph before Theorem 17 in \cite{BGH2}. 
Using the results of Sections \ref{sec2} and \ref{sec:terminal}, we obtain a more concrete version of Corollary \ref{CorWei}, which determines 
%Now here is a description of 
several initial and terminal higher weights of the $\Ch^\perp$. % and 
%It may be noted that this subsumes \cite[Theorem 17]{BGH2}.
As a very special case, we also obtain an alternative and simpler proof of \cite[Theorem 17]{BGH2}.
% of this result. 

\begin{theorem}
\label{dualdr}
Assume that %$h\ge 1$. % and 
$\ell'>1$. %For $1\le r \le k_h$, let $d_r$ denote the $r^{\rm th}$ higher weight of $\Ch$ and for 
$1\le r \le n - k_h$, let $d_r^{\perp}$ denote the $r^{\rm th}$ higher weight of $\Ch^{\perp}$. Then
\begin{enumerate} 
\item[{\rm (i)}] 
%$d_1^{\perp}$ = 3 or 4 according as $q> 2$ or $q =2$. Further, for $2 \le r \le q^{\ell'} - \ell'$,
The minimum distance %$d_1^{\perp}$ 
of $\Ch^{\perp}$ is given by 
$$
d_1\left( \Ch^{\perp} \right) = \begin{cases}
3 & \text{ if $q>2$}, \\
4 & \text{ if $q=2$}.
\end{cases}
$$
More generally, upon letting $Q_j = q^j - j$ for $j\ge 0$, %$0\le j \le \ell'$, 
the $s^{th}$ higher weight of $\Ch^{\perp}$ for $1 \le s < q^{\ell'} - \ell'$ is given by 
$$
%\begin{equation}
%\label{DualChInitial}
d_s\left( \Ch^{\perp} \right) = s+ j+1
%\end{equation}
$$
where $j$ is the unique positive integer $\le \ell'$ such that $Q_{j-1} \le s < Q_j$. 
\item[{\rm (ii)}]
With $d(\ell,m;h) $ as in \eqref{eq:dlmr2}, we have $d(\ell,m;h) \ge 2$ and 
$$
%\begin{equation}
%\label{DualChTerminal}
d_{n-k_h-s}\left( \Ch^{\perp} \right) = q^{\delta} - s \quad \text{ for } 0\le s \le d(\ell,m;h) -2.
%\end{equation}
$$
In particular, $\Ch^{\perp} $ is nondegenerate. Further if we assume 
that $h < \ell'$ or $\ell' > 2h$, and we let $G_0=H_0=0$ and   
%$\mu'$ be as in \eqref{mumu'}. Then $d_{n-k_h}^{\perp} = n = q^{\d}$.
%Further, Also, for $1\le r \le d_{\mu'} - \mu'$, upon letting  
%$G_j$ denote the Griesmer-Wei factor 
$G_j:=\sum_{i=0}^{j-1} q^{-i}$ and %$H_j$ denote 
$H_j:=d(\ell,m;h)G_{j} -j$ for any positive integer $j$, then 
the last $s^{th}$ higher weight of $\Ch^{\perp}$  for 
%$d(\ell,m;h) -1 \le s \le \max\{ d(\ell,m;h)G_{\ell} - \ell , \;  d(\ell,m;h)G_{\ell'-\ell} -(\ell'-\ell)\}$ 
$H_1 \le s \le \max\{H_{\ell}, \, H_{\ell'-\ell} \}$ is given by 
$$
d_{n-k_h-s}\left( \Ch^{\perp} \right) = n- s - j,
$$
where $j$ is the unique positive integer $\le \max\{\ell, \ell' -\ell\}$ %such that 
with $H_{j-1} \le s < H_j$.
%$d(\ell,m;h)G_{j} - j \le s <  d(\ell,m;h)G_{j+1} - (j+1)$. 
\end{enumerate}
\end{theorem}

\begin{proof}
Let $C = \Ch$, $n=q^{\d}$ and $k = k_h$,   and let $d_j, e_j, f_j$ be as in Corollary \ref{CorWei}. By Theorem \ref{ValueOfdr}, the code $C$ is nondegenerate and $d_{k-j} = n - q^{j-1}$ for $1\le j \le \ell'+1$. Consequently, the condition $f_j < s \le  f_{j+1}$ translates 
to $Q_{j-1} \le s < Q_j$, provided $1\le j \le \ell'$. Thus \eqref{DualInitial} implies the desired formula  in (i) for $d_s\left( \Ch^{\perp} \right) $. %, provided 
In the particular case when $s=1$, we have $Q_1=1< 2= Q_2$ or $0=Q_0 < 1 < Q_1$ according as $q=2$  or $q>2$, and this yields the formula for the minimum distance 
of $\Ch^{\perp}$.

Next,  $d_1 = d(\ell,m;h) = q^{\delta - \ell^2}|\GL{\ell}(\Fq)|$ and since 
$\ell'> 1$ we see that $q^{\delta - \ell^2} \ge 2$ when $\ell < \ell'$, whereas
$|\GL_{\ell}(\Fq)| \ge (q^2 -1) (q^2 -q) \ge 6$ when $\ell = \ell'$. Thus in any case $d(\ell,m;h) \ge 2$ and so \eqref{DualTerminal} implies the first assertion in (ii). %\eqref{DualChTerminal}. 
Further, Theorem \ref{mainthm1} shows that $e_j = H_j$ for $1 \le j \le 1 + \max\{\ell, \ell' -\ell\}$. Thus \eqref{DualTerminal} implies the remaining assertion in (ii) as well.
\end{proof}

\begin{remark}
\label{rem:dualdr}
With the first %and the last higher weight of $\Ch^{\perp}$ given as in part~(i) %and (ii) 
higher weight of $\Ch^{\perp}$ given as in part~(i) 
of Theorem \ref{dualdr} above, we can also describe %several initial and terminal 
many of the initial higher weights $d^{\perp}_s : = d_s\left( \Ch^{\perp} \right) $ by the recursive formula 
$$ 
d^{\perp}_s = 
\begin{cases}
d_{s-1}^{\perp} + 2 & \text{if $d_{s-1}$ is a power of $q$}, \\
d_{s-1}^{\perp} + 1 & \text{otherwise},
\end{cases}
$$
provided $2 \le s \le  q^{\ell'} - \ell'$. Likewise,  the last higher weight of $\Ch^{\perp}$ is  $n=q^{\delta}$, and %then several 
many terminal higher weights of 
$\Ch^{\perp}$ are given by the recursive formula 
$$
d^{\perp}_{n - k_h - s} = \begin{cases}
d_{n-k_h-s+1}^{\perp} - 2 & \text{if $d_{n-k_h-s+1} = n + 1 - d(\ell,m;h) G_j$ for some $j$}, \\
d_{n-k_h-s+1}^{\perp} - 1 & \text{otherwise}, \end{cases}
$$
provided $1 \le s \le \max\{ d(\ell,m;h)G_{\ell} - \ell , \;  d(\ell,m;h)G_{\ell'-\ell} -(\ell'-\ell)\}$ 
and it is assumed that $h < \ell'$ or $\ell' > 2h$.  
\end{remark}
 
Using the direct formula in Theorem \ref{dualdr} or the recursive formula in Remark~\ref{rem:dualdr}, we can easily write down several of the initial and terminal higher weights of 
 $\Ch^{\perp}$. 
Table \ref{table:nonlin} illustrates the first few higher weights $d^{\perp}_s$ 
%:=  d_s\left( \Ch^{\perp} \right) $ 
of the dual of $\Ch$, where $h, \ell, \ell'$ are sufficiently large, say $\ell'> \ell \ge h \ge 27$.

%\newpage

\begin{table}[h] 
\caption{Dual Higher Weights of Affine Grassmann Codes }  
\centering 
\begin{tabular}{c | c c c c c c c c c c c }   
\hline\hline  
$q$ & 2 & 3 & 4 & 5 & 7 & 8 & 9 & 11 & 13 & 16 &17 \\ [0.3ex] 
\hline 
$d^{\perp}_1$ &   4 &   3 &    3 &   3 &    3 &   3 &     3 &      3 &     3 &    3 &   3 \\ [.3ex]  
$d^{\perp}_2$	&         6 &	5 &	4 &	4 &	4 &	4 &	4 &	4 &	4 &	4 &	4 \\ [.3ex]
$d^{\perp}_3$ & 	7 &	6 &	6 &	5 &	5 &	5 &	5 &	5 &	5 &	5 &	5 \\ [.3ex]
$d^{\perp}_4$ &	8 &	7 &	7 &	7 &	6 &	6 &	6 &	6 &	6 &	6 &	6 \\ [.3ex]
$d^{\perp}_5$ &	10 &	8 &	8 &	8 &	7 &	7 &	7 &	7 &	7 &	7 &	7 \\ [.3ex]
$d^{\perp}_6$ &	11 &	9 &	9 &	9 &	9 &	8 &	8 &	8 &	8 &	8 &	8 \\ [.3ex]
$d^{\perp}_7$	&       12 &	11 &	10 &	10 &	10 &	10 &	9 &	9 &	9 &	9 &	9 \\ [.3ex]
$d^{\perp}_8$ &	13 &	12 &	11 &	11 &	11 &	11 &	11 &	10 &	10 &	10 &	10 \\ [.3ex]
$d^{\perp}_9$ &	14 &	13 &	12 &	12 &	12 &	12 &	12 &	11 &	11 &	11 &	11 \\ [.3ex]
$d^{\perp}_{10}$ &	15 &	14 &	13 &	13 &	13 &	13 &	13 &	13 &	12 &	12 &	12 \\ [.3ex]
$d^{\perp}_{11}$ &      16 &	15 &	14 &	14 &	14 &	14 &	14 &	14 &	13 &	13 &	13 \\ [.3ex]
$d^{\perp}_{12}$ &	18 &	16 &	15 &	15 &	15 &	15 &	15 &	15 &	15 &	14 &	14 \\ [.3ex]
$d^{\perp}_{13}$ &	19 &	17 &	16 &	16 &	16 &	16 &	16 &	16 &	16 &	15 &	15 \\ [.3ex]
$d^{\perp}_{14}$ &	20 &	18  &	18 &	17 &	17 &	17 &	17 &	17 &	17 &	16  &	16 \\ [.3ex]
$d^{\perp}_{15}$ &	21 &	19 &	19 &	18 &	18 &	18 &	18 &	18 &	18 &	18 &	17 \\ [.3ex]
$d^{\perp}_{16}$ &	22 &	20 &	20 &	19 &	19 &	19 &	19 &	19 &	19 &	19 &	19 \\ [.3ex]
$d^{\perp}_{17}$ &	23 &	21 &	21 &	20 &	20 &	20 &	20 &	20 &	20 &	20 &	20 \\ [.3ex] 
$d^{\perp}_{18}$ &	24 &	22 &	22 &	21 &	21 &	21 &	21 &	21 &	21 &	21 &	21 \\ [.3ex]  
$d^{\perp}_{19}$ &	25 &	23 &	23 &	22 &	22 &	22 &	22 &	22 &	22 &	22 &	22 \\ [.3ex]
$d^{\perp}_{20}$ &	26 &	24 &	24 &	23 &	23 &	23 &	23 &	23 &	23 &	23 &	23 \\ [.3ex]
$d^{\perp}_{21}$ &	27 &	25 &	25 &	24 &	24 &	24 &	24 &	24 &	24 &	24 &	24 \\ [.3ex]
$d^{\perp}_{22}$ &	28 &	26 &	26 &	25 &	25 &	25 &	25 &	25 &	25 &	25 &	25 \\ [.3ex]
$d^{\perp}_{23}$ &	29 &	27 &	27 &	27 &	26 &	26 &	26 &	26 &	26 &	26 &	26 \\ [.3ex]
$d^{\perp}_{24}$ &	30 &	29 &	28 &	28 &	27 &	27 &	27 &	27 &	27 &	27 &	27 \\ [.3ex]
$d^{\perp}_{25}$ &	31 &	30 &	29 &	29 &	28 &	28 &	28 &	28 &	28 &	28 &	28 \\ [.3ex] 
$d^{\perp}_{26}$ &	32 &	31 &	30 &	30 &	29 &	29 &	29 &	29 &	29 &	29 &	29 \\[.3 ex]
$d^{\perp}_{27}$ &	34 &	32 &	31 &	31 &	30 &	30 &	30 &	30 &	30 &	30 &	30 \\ [.3 ex]
%\\ [1ex]
\hline  
\end{tabular} 
\label{table:nonlin} % is used to refer this table in the text 
\end{table} 

\appendix
\renewcommand\thesection{\!\!}
\section{A Geometric Approach to Higher Weights}

Let $n,k$ be positive integers with $k\le n$. A nondegenerate $[n,k]_q$-projective system is simply a (multi)set $X$ of $n$ points in the projective space $\PP^{k-1}$ %= \PP^{k-1}(\Fq)$.
over the finite field $\Fq$. 
If we write $\PP^{k-1} = \PP(V)$, where $V$ is a $k$-dimensional vector space over $\Fq$ and fix some lifts, say $v_1, \dots , v_n$, of these $n$ points to $V$, then the associated nondegenerate linear code $C_X$ is the image of the evaluation map 
$$
\Ev: V^* \to \Fq^n \quad \text{ defined by } \quad \Ev (\phi) = \left( \phi(v_1), \dots , \phi (v_n) \right),
$$
where $V^*$ denotes the dual of $V$, i.e., the space of all linear maps from $V$ to $\Fq$.  It is shown in \cite{TV1,TV2} that the association $X \leadsto C_X$ is a one-to-one correspondence, modulo natural notions of equivalence, from the class of nondegenerate $[n,k]_q$-projective systems onto the class of nondegenerate $[n,k]_q$-codes. For $1\le r \le k$, the $r^{\rm th}$ higher weight of $C_X$ corresponds to maximal sections of $X$ by (projective) linear subspaces %varieties 
of $\PP^{k-1}$ of codimension $r$; more precisely,
$$
d_r(C_X) = n - \max\{ |X\cap \Pi| : \Pi \text{ linear subspace of $\PP^{k-1}$ with } \codim \Pi = r\}.
$$
For more on this, we refer to \cite{TV1,TV2}. 

Now let $\ell, m$ be positive integers with $\ell \le m$ and as before let $k={{m}\choose{\ell}}$, 
$\ell':= m-\ell$ and $\d = \ell \ell'$. Assume that $1< \ell \le \ell'$. Consider the Grassmannian $G_{\ell,m} = G_{\ell,m} (\Fq)$ of $\ell$-dimensional subspaces of $\Fq^m$. The Pl\"ucker embedding 
$$
G_{\ell,m} \hookrightarrow \PP^{{{m}\choose{\ell}}-1} = \PP(\wedge^{\ell}\Fq^m) \quad \text{ given by } \quad W = \langle w_1, \dots , w_{\ell}\rangle \mapsto [w_1\wedge \cdots \wedge w_{\ell}]
$$
is known to be nondegenerate and the corresponding nondegenerate linear code is the Grassmann code $C(\ell,m)$. The formula stated in the Introduction for the last few higher weights of $C(\ell,m)$ follows readily from the structure of linear subvarieties of $G_{\ell,m}$ or, in algebraic parlance, the structure of decomposable subspaces of exterior powers. Indeed, $G_{\ell,m}$ contains a linear subspace $\Pi$ of dimension $r-1$, provided $r \le \mu$, where  $\mu:=1+\max\{\ell,m-\ell\}$; see, for example, \cite[Cor. 7]{GPP} or \cite[Lemma 3.5]{GK}. 
This subspace $\Pi$ has codimension $(k-1)-(r-1) = k-r$ in $\PP^{k-1}$, and clearly, $|\Pi\cap G_{\ell,m}| = |\Pi| = |\PP^{r-1}| = (1+q+ \cdots + q^{r-1})$. 
Consequently, 
$$
d_{k-r}\left( C(\ell ,m)\right) = n -( 1+ q + \cdots + q^{r -1}) \quad \text{ for } 1\le r \le \mu.
$$
Suppose we fix an ordered basis $\{e_1, \dots , e_m\}$ of $\Fq^m$ and the corresponding basis $\{e_{\alpha} : \alpha\in I(\ell,m)\}$ of $\wedge^{\ell}\Fq^m$, where 
$$
I(\ell ,m)=\{ \a = (\a_1, \dots , \a_\ell ) : \in {\mathbb Z}^{\ell}  : 
 %\ \a_1, \dots , \a_\ell  \ \mbox{are integers such that } 
1\le \a_1 < \dots < \a_\ell  \le m \} 
$$
and $ e_{\alpha} : = e_{\alpha_1}\wedge \cdots \wedge e_{\alpha_\ell}$ for $\a = (\a_1, \dots , \a_\ell ) \in I(\ell, m)$. The Pl\"ucker coordinates of an $\ell$-dimensional subspace $W \in G_{\ell,m}$ spanned by $\{ w_1, \dots , w_{\ell}\}$ are precisely $p = (p_{\a})_{\a\in I(\ell,m)}$, where $p_{\a}\in \Fq$ are %uniquely 
determined by the relation 
$$
w_1\wedge \cdots \wedge w_{\ell} = \sum_{\a\in I(\ell,m)} p_{\a}e_{\a}.
$$ 
For $\a\in I(\ell,m)$, let $H_{\a}$ denote the hyperplane $\{p \in \PP^{k-1} : p_{\a} =0\}$ in 
$\PP^{k-1} $, and let $U_{\a}:=\{p \in \PP^{k-1} : p_{\a} \ne 0\}$ be the corresponding basic open set. % isomorphic to $\AA^{k-1}$. 
It is a classical fact that  $U_{\a}\cap G_{\ell,m}$ %\cap \{p \in \PP^{k-1} : p_{\a} = 1\}$ 
is isomorphic to the affine space $\AA^{\d}$ of $\ell \times\ell'$ matrices over $\Fq$.
% for any $\a\in I(\ell,m)$. 
This correspondence is given explicitly by the Basic Cell Lemma of \cite{GL}. For the sake of definitiveness, consider $\theta:= (\ell'+1, \ell'+2, \dots , m)\in I(\ell,m)$. Then the Pl\"ucker embedding restricted to $U_{\theta}\cap G_{\ell,m}$ gives a nondegenerate embedding of $\AA^{\d}$ into $\PP^{k-1} $, and the linear code corresponding to this projective system is, in fact, the affine Grassmann code $C^{\AA}(\ell,m)$ of length $q^{\d}$. If $\Pi$ is a linear subspace of $\PP^{k-1} $ of dimension $r-1$, then $\Pi\cap H_{\theta}$ would be a linear subspace of dimension $r-2$ or $r-1$ according as $\Pi\not\subseteq H_{\theta}$ or $\Pi\subseteq H_{\theta}$. Consequently, %and hence 
$$
|\Pi \cap U_{\theta}| = |\Pi| - |\Pi\cap H_{\theta}| \le ( 1+ q + \cdots + q^{r -1}) - ( 1+ q + \cdots + q^{r -2}) = q^{r -1}.
$$
Consequently, $|\Pi \cap U_{\theta}\cap G_{\ell,m}| \le q^{r -1}$ and so 
$d_{k-r}(\C) \ge q^{\delta} - q^{r- 1}$ for all $r=1, \dots , k$. This proves a stronger version of Lemma \ref{BoundOndr} in the case $h=\ell$. Further, if $r\le \mu = \ell'+1$ and if $\Pi$ is a linear subspace of $\PP^{k-1} $ of codimension $k-r$ chosen in such a way that $\Pi\subseteq G_{\ell,m}$ and $\Pi\subseteq H_{\theta}$, then 
$|\Pi \cap U_{\theta}\cap G_{\ell,m}| = |\Pi\cap H_{\theta}| = q^{r -1}$ and so 
$d_{k-r}(\C) = q^{\delta} - q^{r- 1}$ for all $r=1, \dots , \mu$. Since $1<\ell \le \ell'$, choosing such a subspace $\Pi$ is possible for $1\le r\le \mu$; for example, we can take 
$\Pi =\{p \in \PP^{k-1} : p_{(1, 2, \dots , \ell-1, j)} =0 \text{ for } j=\ell, \ell+1, \dots , \ell + r-1\}$ to be the intersection of Pl\"ucker coordinate hyperplanes that are ``close'' to each other. Thus we obtain an alternative proof of Theorem \ref{ValueOfdr} when $h=\ell$. 
 
On the other hand, deriving the formulas that we have for initial higher weights of $\C$ from the corresponding results for the Grassmann code $C(\ell, m)$ is not so straightforward. To be sure, the optimal linear subspace in $G_{\ell,m}$ of large dimension (or small codimension) are obtained in \cite{GL} by considering close families in $I(\ell,m)$ and the corresponding linear subsbaces of $\PP^{k-1}$ given by the intersections of Pl\"ucker coordinate hyperplanes. Recall that 
$\Lambda \subseteq I(\ell,m)$ is said to be \emph{close} if any two distinct elements  of $\Lambda$ have $\ell-1$ coordinates in common. However, determining the maximum possible  cardinality of the intersection of the corresponding linear subspace $\Pi$ with $U_{\theta}\cap G_{\ell,m}$ is not easy. It may be tempting to consider  $\Lambda \subseteq I(\ell,m)$ not containing $\theta$ such that $\Lambda \cup \{\theta\}$ is close. But this doesn't work even when $\Lambda$ is singleton (which would correspond to looking at the minimum distance). In fact, it is better to keep the elements of $\Lambda$ as far away from $\theta$ as possible. Thus choosing a close family in $I(\ell, \ell')$ rather than $I(\ell,m)$ is helpful and this has, in fact, motivated the proofs of Lemma \ref{closeminors} and \ref{closeminors2}, which paved the way for Theorem \ref{mainthm1}.

\end{document}